\documentclass[a4paper,12]{article}
\usepackage{amsmath,amsfonts,amssymb,graphicx,cite}
\usepackage[paper=letterpaper,margin=1.2in]{geometry}
\usepackage{booktabs}
\usepackage{amsthm}
\usepackage{multirow}
\usepackage{sidecap}
\usepackage{verbatim}
\usepackage{wasysym}
\usepackage{appendix}

\newtheorem{proposition}{\bf PROPOSITION}

\newcommand{\setall}{\setcounter{equation}{0}\setcounter{theorem}{0}\setcounter{table}{0}\setcounter{footnote}{0}}

\def\nn{\nonumber}

\begin{document}

\begin{titlepage}
~\\
\vskip 1cm
\centerline{\Large{\bf Probing the Space of Toric Quiver Theories}}
\vskip 1cm
\centerline{{\large
Joseph Hewlett$^{1,2}$ \footnote{\tt joe.cyn@gmail.com}
and
Yang-Hui He${}^{1,3}$ \footnote{\tt hey@maths.ox.ac.uk}
}}
~\\
~\\
{\hspace{-1in}
\scriptsize
\begin{tabular}{ll}
  ${}^1$
  & {\it Rudolf Peierls Centre for Theoretical Physics, Oxford University, 1 Keble Road, OX1 3NP, U.K.}\\
  ${}^2$
  & {\it Perimeter Institute for Theoretical Physics,31 Caroline St. N. Waterloo, Ontario, Canada, N2L 2Y5}\\
  ${}^3$
  & {\it Collegium Mertonense in Academia Oxoniensi, Oxford, OX1 4JD, U.K.}\\
\end{tabular}
}

\vskip 2cm

\begin{abstract}
We demonstrate a practical and efficient method for generating toric Calabi-Yau quiver theories, applicable to both D3 and M2 brane world-volume physics.
A new analytic method is presented at low order parametres and an algorithm for the general case is developed which has polynomial complexity in the number of edges in the quiver. 
Using this algorithm, carefully implemented, we classify the quiver diagram and assign possible superpotentials for various small values of the number of edges and nodes.
We examine some preliminary statistics on this space of toric quiver theories. 
\end{abstract}

\end{titlepage}

\tableofcontents

\setall
\section{Introduction}
The works of \cite{BL,gus} have brought about tremendous interest in the M2-brane, and following \cite{Aharony:2008ug}, where the world-volume theory was realised as a Chern-Simons quiver theory in $(2+1)$-dimensions, much activity ensued (q.v.~e.g.~\cite{Lambert:2008et,Franco:2008um,Benna:2008zy,Hosomichi:2008jb,Hanany:2008qc,Ueda:2008hx,Imamura:2008qs,Martelli:2008rt,Martelli:2008si,Hanany:2008cd,Hanany:2008fj,Davey:2009sr,Amariti:2009rb,Hanany:2009vx,Aganagic:2009zk,Davey:2009qx,BHS,Hanany:2008gx,talk,Papathanasiou:2009en}). Even though nice parallels are drawn between the familiar case of D3-branes probing Calabi-Yau threefold singularities and the present circumstance of M2-branes probing Calabi-Yau fourfold singularities, the latter situation is far less understood. Complications arise in the various analogues of the wealth of properties enjoyed by the D3-brane, such as singularity-resolution in relation to (un)Higgsing, the ``inverse algorithm'' for systematically constructing the world-volume gauge theory,  Seiberg duality, etc. Part of the issue arises from the new possibility of turning on G-fluxes on torsion-cycles in the dual $AdS$ geometry \cite{Aganagic:2009zk,BHS,talk}.

Confronted with these seemingly untamable complexities, an optimistic path, guided by brane configurations, had been trodden. It is by now well-established that the $(3+1)$-dimensional supersymmetric quiver gauge theory of D3-brane probing a toric Calabi-Yau singularity is completely described by a planar, periodic tiling of NS5-branes and D5-branes, or, equivalently, by an auxiliary bi-partite graph on a torus known as a dimer model.  It is suspected that this persists to the case of our $(2+1)$-dimensional quiver Chern-Simons theories, a conjecture fortified by the observation that almost all such theories discovered to date, being of the M2-brane world-volume, descend from parent D3-brane theories. More precisely, take the quiver and superpotential of a world-volume theory arising from a D3-brane probing a Calabi-Yau threefold, add Chern-Simons levels and terms respectively to the graph labeling and the interaction, the new moduli space is computable to be a Calabi-Yau fourfold, apparently rather different, but is expected to be the space which a M2-brane probes. This ``forward algorithm'' is succinctly presented in \cite{Hanany:2009vx}.

Led by this beacon, a systematic, taxonomic study of toric quivers was undertaken in \cite{Hanany:2009vx} and it is in this spirit that we wish to proceed with our current investigations. Our purpose is twofold. First, we wish to continue with \cite{Hanany:2009vx} and exhaustively and progressively classify toric quiver diagrams with superpotential, whereby establishing, for each given number of nodes and edges, all possible parent $(3+1)$-dimensional theories which could descend to $(2+1)$-dimensional Chern-Simon quiver theories. We will see this is facilitated by a tremendous improvement of generating quivers, using an efficient notation and algorithm. Second, we would like to take advantange of this new tool and attempt an exploration of the space of quiver theories along the vein of \cite{Gray:2005sr,Gray:2008yu}, and gather data for our experimental probe of the string geometric landscape.

The paper is organised as follows. 
In the next section we briefly review some topics in graph theory, complexity and toric Calabi-Yau quiver gauge theories. 
This is followed by a restatement of the problem, rephraseded in this language. 
Section 3 offers a new method by which toric Calabi-Yau quiver theories can be generated, a new notation for specifying quiver diagrams and an example of both of these in the form of the complete generation of a low order set of quiver diagrams. 
Section 4 presents an algorithm used for generating complete sets of quiver theories. 
In section 5 we show our results, presenting some complete sets of quiver diagrams and associated superpotentials; as well as some statistics of the total numbers inequivalent quiver theories for some low order parametres. 
Finally, in section 6 we conclude with a discussion of the work, and its futher prospects.

We would like to draw the reader's attention to a beautiful parallel work by \cite{amisteam}, which also classifies toric quiver theories, but from the point of view of brane tilings. 
The two works thus complement each other harmoniously.

\setall
\section{The Taxonomical Problem}
\label{Sec:TP}
We start with introducing the concept of a toric Calabi-Yau quiver theory (TCYQT) in an abstract sense.
We refer the reader to \cite{Douglas:1996sw}, which first brought the study of quiver theories to gauge theories, for details, as well as to \cite{He:2004rn} for a rudimentary review.

\subsection{Some Rudiments on Graphs and Complexity}
First, let us recall some elementary facts on graph theory and algorithmic complexity, which will be crucial to our constructions later.
A \emph{graph} is an ordered pair $(G,E)$, where $G$ is a set of nodes and $E$ is a set of pairs of nodes, which represent the two ends of an edge. In a directed graph, $E$ contains ordered pairs, which tell us the edge direction\cite{hartmann-2005}. Node $i$ is \emph{adjacent} to node $j$ if the directed edge $(i,j)$ is in the graph.

A \emph{path} on a directed graph with nodes $G=\{n_{i},i=\{1..p\}\}$ is an ordered set of edges $\{E=\{e_{j},j=\{1..q\}\}$ where $e_{j} = (n_{k},n_{j})$ and $e_{j+1} = (n_{j},n_{l})$. That is, we follow the arrows on the graph between two given nodes $k \rightarrow l$. 
In a directed graph, each node has an `\emph{indegree}' and `\emph{outdegree}', which are the number of edges entering and leaving the node respectively. 
A \emph{connected} graph has a path between any two nodes and a \emph{cycle} is a path which ends at its starting node.
A graph is \emph{Eulerian} if there exists at least one cycle which visits every edge once and only once \cite{hartmann-2005}. 
A quiver graph can be viewed as an extension of a directed graph and is the quadruple $(G,E,d,R)$ where $d$ is a set of node labels which tell us the gauge groups of the theory and $R$ is a set of algebraic labels for the edges. The choice of $d$ and $R$ is known as the \emph{quiver representation} \cite{quiverrep}.

A major motivation of this investigation is to establish an efficient method of explicitly constructing all toric quivers. 
The best way to analyse the quality of an algorithmic method is to examine its \emph{complexity}, i.e., how the time the algorithm takes to run increases with the size of its arguments.
An algorithm has polynomial complexity \cite{500824} if its running time increases as some power of the size of the argument. 
For example, looping through a list has complexity $O(N)$, where $N$ is the number of terms in the list. 
As a general rule, algorithms with polynomial complexity are known as ``fast'' algorithms and ones with exponential complexity are ``slow'' \cite{500824}. 

There is an important area where graph theory and computer science meet and that is the question of \emph{graph isomorphism}. 
Two graphs are said to be isomorphic, and so belong to the same \emph{graph isomorphism class} if there exists a relabeling of the nodes which makes the graphs identical. 
Currently there is no known simple property of a graph which can be used to check whether it is isomorphic to any other. 
It is an open question as to how efficient an algorithmic solution for checking graph isomorphism can be. 
The current quickest solution to the isomorphism problem is found in \cite{1382666}. For a detailed discussion of the problem see\cite{GIsoDisease}.

\subsection{Toric Calabi-Yau Quiver Gauge Theories}
Prepared with the above elements, we now introduce quiver graphs in the context of brane world-volume gauge theories.
A quiver theory is a pair $(Q,W)$ where $Q$ is a finite {\bf quiver diagram} (directed, labeled and represented in the aforementioned sense) and $W$ is a formal polynomial in the arrows of $Q$, called the {\bf superpotential}.
The finiteness signifies that there is a finite number of nodes, which are to correspond to gauge factors in a gauge theory, and a finite number of directed arrows, which are to correspond to bi-fundamental fields.
The labelings are integers to be assigned to the nodes, and are to denote the rank of the gauge factors as well as possible Chern-Simons levels.
Let $Q$ have $G$ nodes, $E$ arrows and let $W$ have $N_T$ terms.
The arrows will be denoted as $X_{ab}^i$, signifying fields charged under nodes $a$ and $b$, with possible multiple arrows indexed by $i$.
The superpotential is then a polynomial in these fields, each monomial term is a gauge invariant and is thus a product over closed loops in the quiver; furthermore, we require the monomials to be at least cubic since quadratic terms are mass terms which could be integrated out.
For now, we shall take all fields $X_{ab}^i$ to be simply complex numbers, this will be the first of so-called toric conditions, and in particular corresponds to the case of a single brane-probe.
In general, we should promote these to matrices, in which case there are multiple parallel coincident branes.

We shall represent $Q$ by the {\bf incidence matrix} $d$, which is a $G \times E$ matrix of $\pm 1$ and $0$ entries: each row enumerates a node and each column, an arrow, such that an arrow from node $a$ to $b$ will have, in the column representing this arrow, the $a$-th row being $-1$ and the $b$-th row, $+1$; all other entries are 0.
Clearly, the sum over all the rows gives a row of zeros by definition.
Note that this matrix does not capture so-called adjoint fields which go from a node to itself.

To each such a theory $(Q,W)$ one can construct a so-called vacuum moduli space (VMS).
This, by construction, is the geometry which a brane should probe.
The procedure by which one finds the VMS from the pair $(Q,W)$ has been called the ``Forward Algorithm'' \cite{Feng:2000mi}.
In the particular case where the VMS is a toric Calabi-Yau space, this algorithm can be succinctly summarised as manipulation of integer matrices and combinatorics (see the flow-charts on page 7 of \cite{Hanany:2008gx}).
The reverse problem, of geometrically engineering the brane world-volume theory given the vacuum geometry is likewise called the ``Inverse Algorithm''.

Having introduced the above fundamentals of the quiver theories which we will systematically study in this paper, we now move on to placing some mild constraints so as to restrict ourselves to a sub-class which had been of primary interest over the last decade, viz., toric Calabi-Yau conditions.
As we mentioned above, we shall first take all gauge groups to be $U(1)$, whence all fields are complex valued.
Next, we shall require that the VMS, computed from the Forward Algorithm, be Calabi-Yau, this requires that the sum over the columns also gives a column of zeros.
Now, in $(3+1)$-dimensions, this is a consequence of anomaly-cancelation: that for each node, the indegree and outdegree should equal.
Note that this is a sufficient but not necessary condition for the VMS to be Calabi-Yau, however, for simplicity we will henceforth impose it on all our quivers.

In order that the VMS be {\it toric} Calabi-Yau, further constraints are to be met by the superpotential. First, as pointed out in \cite{Feng:2000mi} and dubbed the ``toric condition'', each field must appear {\it exactly } twice and with opposite sign. Second, it is now well-established that toric Calabi-Yau brane theories admit a tiling or dimer description (see \cite{dimer} for excellent reviews); in order to admit a tiling, we must also adhere to a topological condition
\begin{equation}\label{tile}
G - E + N_t = 0  \ .
\end{equation}

\subsection{Summary of Constraints}
We are thus confronted by a classification problem.
Our strategy will be to  generate incidence matrices of size $G \times E$ which represent the quiver diagrams, such that the points 1 and 2 of the following set of conditions are met.
To each such a quiver $Q$ we must generate all possible superpotentials $W$ which satisfy conditions 3 to 6:
\begin{enumerate}
\item Every column in the incidence matrix must sum to zero (incidence condition).
\item Every row in the incidence matrix must sum to zero (Calabi-Yau condition). 
\item Every superpotential must satisfy the toric condition.
\item Every superpotential must contain $x = E-G$ terms (tiling condition).
\item Every term in the superpotential must correspond to a loop in the quiver diagram (gauge invariance).
\item Every term in the superpotential must contain at least three fields.
\end{enumerate}

Computationally, we may proceed as follows.
The initial space of possible matrices are those which contain entries that are $+1$, $-1$ and $0$. Conditions 1 and 2 should be imposed on them, producing a quiver diagram incidence matrix and reducing the space of possible incidence matrices. 
Then  a representative of each graph isomorphism class is required from this reduced set.
The brute force approach to generating possible incidence matrices has at worst complexity $O(3^{E(E-x)})$, as each element in the matrix can be one of three entries.
Next, we also wish to find a constructive way of forming superpotentials. 
The strategy will be to find a way of producing a superpotential using the given quiver as a foundation. 
Here, the brute force algorithm for superpotential generation has complexity $O(2^{2^x})$, as the set of all sets of monomial terms is required.
Clearly, the brute-force method is not feasible and from a computational perspective a fast algorithm is desireable, though due to the difficulty of the problem the emphasis is on finding a simple workable solution, easily testable to eliminate errors. 
We will now devise such an algorithm.

\setall
\section{New Notations and Generations for Toric Quivers}
\label{Sec:Theory}
As we saw above, we need an efficient, or at the least, a practical method of systematically generating toric quiver theories.
Preliminary attempts were made, using standard partition routines, especially implemented on Mathematica$^{\mbox{\scriptsize\textregistered}}$, in \cite{Hanany:2008gx}.
However, there, the taxonomical study was already hindered by the growth rate of the complexity.
In this section, we will start afresh, and introduce a new notation for generating quivers, engendering all from a polygonal shape, and then rather effectively enumerate our desired theories.
\subsection{Collapsing Polygons and Constructing Superpotentials}
We begin with the simple observation that every toric quiver diagram admits at least one Eulerian cycle since every node has equal indegree and outdegree due to the Calabi-Yau condition. 
The original proof of this is due to Fleury\cite{hartmann-2005} and a proof is included in Appendix \ref{App:QP}. 
This realisation that a quiver diagram admits an Eulerian cycle gives a route to a smaller space of possible toric incidence matrices. 
Consider a polygon, with nodes at the corners and clockwise directed edges. 
By pinching together two nodes in the polygon, it is collapsed to a graph with fewer nodes. 
This graph clearly obeys all conditions for a toric Calabi-Yau quiver diagram. 

This above intuition can be formalised.
Suppose we wish to study a toric quiver with $G$ nodes and $E$ arrows, we define $x = E - G$ as an order parametre as was done in \cite{Hanany:2008gx}; this of course, will turn out to be the number of terms $N_t$ in the superpotential by condition \eqref{tile}.
Now, we can define a graph operation $C(i,j)$ on the polygonal quiver as one which pinches the graph between nodes $i$ and $j$. 
After this is applied, the resulting node is labeled simultaneously with both $i$ and $j$. 
The polygonal quiver together with this pinching operation is shown diagrammatically in Figure \ref{f:pinch}.
\begin{figure}[h]
\centering\includegraphics[width=5in]{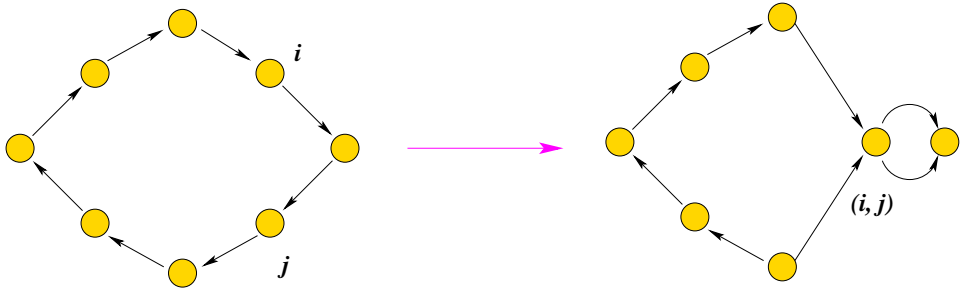}
\caption{{\sf Pinching a cyclically directed polygonal graph by the operation $C(i,j)$.}}
\label{f:pinch}
\end{figure}

It turns out that the above simple concepts already suffice to give us a powerful result.
\begin{proposition}
By drawing an $E$-sided directed polygon with $G$ numbered vertices and operating with $C(i,j)$ $x = E-G$ times, for different pairs $(i,j)$, every unique directed toric quiver diagram is generated.
\end{proposition}
\begin{proof}:
Let the set of toric quiver diagrams with $E-G = x$ be $Qu^{x}$, the set of directed polygon graphs be $P$ and the set of graphs formed by acting on $P$ with $C(i,j)$ $x$ times be denoted $CP^{x}$.
We now prove that $Qu^x$ equals $CP^x$.
In one direction, consider an $E$-sided directed polygon. If $C(i,j)$ is applied to nodes $i$ and $j$ the resultant node has two edges entering and two edges leaving the combined node which we denote as $(i,j)$. Repeated application of $C((i,j),k)$ on the resultant node adds another edge entering and leaving each time. So the graph still satisfies the Calabi-Yau condition. 
Applying $C(i,j)$ $x$ times with different $(i,j)$ constructs a directed graph which satisfies the Calabi-Yau condition and still prodces toric quiver diagrams, with $E$ edges and $E-x = G$ nodes. Therefore $CP^{x} \subseteq Qu^{x}$.

Next, consider a quiver, and walk an Eulerian cycle through it. Every time we reach a node with multiple edges entering and leaving we split the node. 
The edges we travel along are now the only edges attached to the new node. 
This is the inverse operation to $C(i,j)$. 
This process forms a directed graph where $E=G$, and constitutes a member of $P$. 
Thus, $Qu^{x} \subseteq CP^{x}$.
Therefore,
\begin{equation}\nn
 CP^{x} \subseteq Qu^{x} \text{ }\&\text{ } Qu^{x} \subseteq CP^{x} 
\Rightarrow  Qu^{x} \equiv CP^{x} 
\end{equation}
\qedhere
\end{proof}
A map $\pi:P \mapsto Qu^{x}$ can now be defined as a way of producing toric quivers. 
This is a way of guaranteeing a very large reduction in the space of possible incidence matrices.

Having constructed the toric quiver $Q$, we now assign the superpotential $W$.
\begin{proposition}
Toric superpotentials can be constructed by decomposing the quiver diagram into disjoint loops to generate terms and then matching up the collections of terms so that the total number of terms is $x$.
\end{proposition}
\begin{proof}
Drawing all loops which correspond to terms of one sign on top of the quiver diagram creates a collection of loops which contains every edge once. 
Thus every set of disjoint loops which covers the quiver diagram represents a collection of terms.
We can then take two collections, one with $n$ terms and one with $x-n$ and subtract one from the other. 
On the condition that the two have no common terms we have made a superpotential with $x$ terms. 
Doing this for all possible combinations of term collections will generate all toric superpotentials. 
\qedhere
\end{proof}
The term collections are constructed by finding all ways of splitting the quiver diagram up into disjoint cycles. 
Consider walking through the quiver, there is a choice of where to go only at the contracted nodes.
We define the concept of \emph{edge assignments} at each node by a set of rules $EA_{i} = \{(E_{a},E_{b}): E_{a} \text{ enters $i$ and }E_{b} \text{ leaves $i$}\}$. 
By constructing all $EA_{i} \text{ }\forall \text{ $i$: indegree($i$)} > 1$ we produce all sets of disjoint loops of the quiver diagram. 

To illustrate, let us use a concrete example.
Let us study the $(G,E) = (3,6)$ quiver diagram.
Note that we will henceforth use this pair notation to denote quivers.
We have drawn the $(3,6)$ quiver in Figure \ref{Qu36S} and now will produce all toric superpotentials for it.
Of course, this is only a demonstrative example, since there will be $x= 6-3=3$ terms in the superpotential and we will see there are non-minimal cycles.
In the actual physics, we may need to promote the node ranks to non-Abelian and shuffle traces amongst the operators.
\begin{figure}[h!]
\centering  \includegraphics[width=0.4\textwidth]{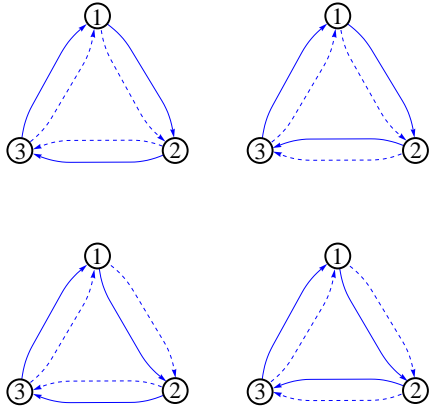}
\caption{{\sf For the $(G,E) = (3,6)$ quiver, we construct all two term collections. One term is the product of the fields in the dashed loop, the other is the product of fields in the solid loop. These will then appear in the superpotential, with opposite signs.}}
\label{Qu36S}
\end{figure}
All sets of two loops that can be used to form the gauge invariant terms are shown in 
Figure \ref{Qu36S}. 
As well as these loops, we also have the term made from an Eulerian cycle of the quiver, namely
$T_{all} = X_{12}^{(1)}X_{12}^{(2)}X_{23}^{(1)}X_{23}^{(2)}X_{31}^{(1)}X_{31}^{(2)}$, looped over all the arrows.
The term collections are then matched up to give $x=3$ terms and hence construct all superpotentials:
\begin{eqnarray}
W_{1}&=&X_{12}^{(1)}X_{23}^{(1)}X_{31}^{(1)}+X_{12}^{(2)}X_{23}^{(2)}X_{31}^{(2)}- T_{all} \nonumber \\
W_{2}&=&X_{12}^{(1)}X_{23}^{(2)}X_{31}^{(1)}+X_{12}^{(2)}X_{23}^{(1)}X_{31}^{(2)}- T_{all} \nonumber \\
W_{3}&=&X_{12}^{(2)}X_{23}^{(1)}X_{31}^{(1)}+X_{12}^{(1)}X_{23}^{(2)}X_{31}^{(2)}- T_{all} \nn \\
W_{4}&=&X_{12}^{(2)}X_{23}^{(2)}X_{31}^{(1)}+X_{12}^{(1)}X_{23}^{(1)}X_{31}^{(2)}- T_{all} \ .
\end{eqnarray}

Indeed, the above combinations are exhaustive and we will perform similar such partitions for all the toric quivers which will be constructed later.
Now, a limit on the quivers which admit toric superpotentials can also be found. 
The quiver must be able to be decomposed into at least $\text{ceil}(\frac{x}{2})$ loops
(we use the standard notation that ceil($x$) rounds $x$ up to the nearest integer and floor($x$) rounds $x$ down), as this is the minimum number of terms required from one term collection to give $x$ terms. Given that there must be at least three fields in each term, the maximum number of terms is $E/3$. Thus,
\begin{equation}
\label{Eqn:LL}
E \geq 3\text{ceil}(\frac{x}{2}) \ .
\end{equation}

\subsection{AB notation}
Now, we introduced the map $\pi:P \mapsto Qu^{x}$ above which explicitly constructs toric quiver diagrams from polygonal contractions. 
To specify the maps we propose a notation which depends as little as possible on labeling the nodes and which reflects the symmetry properties of the quiver. 
This will turn out to be important for redundancy elimination.

We take a walk through a member of $P$ to build the quiver and make the following definitions:
\begin{itemize}
\item Let $C(a) := C(i,i+a)$ where $i$ is the node at which we are currently stopping and $a$ is the number of steps ahead.
\item Let $W(b)$ be the operation of moving along $b$ edges through the nodes in their labeled order from the current position. 
\item We call the process $C(a)$ which contracts two nodes which have already been contracted as a \emph{redundant} contraction.
\item We call a node which has been contracted with another as a \emph{c-node}.
\end{itemize}
We can now define our walk as follows.
Due to the nodes being all labeled by 1 in the toric Abelian case, the walk can begin at any node that will be contracted. 
First, apply $C(a_{1})$ and then $W(b_{1})$, where $a_{1}$ is the number of edges in the first contraction and $b_{1}$ is the number of edges to the next node we wish to contract. 
Then, apply $C(a_{2})$, followed by $W(b_{2})$, and so on. 
Stop once $x$ contractions have been performed. 
Finally apply $W(b_{x})$ to return to the starting node. 
Thus any quiver diagram can be represented by two ordered sets of numbers:
\begin{eqnarray}
\nonumber
A &=& \{a_{i}: i={1,..,x}\}\\
B &=& \{b_{i}: i={1,..,x}\} \ ,
\end{eqnarray} 
where the $a_{i}$ in the $A$ define the contractions in the order in which they are applied, and the $b_{i}$ in the $B$ define the walks between contractions. 
Therefore, a quiver can be described by applying elements from the $A$ and $B$ in turn. That is, using $a_{1}$ followed by $b_{1}$ then $a_{2}$ and so on. 
Thus a general AB is a constructive map:
\begin{displaymath}
(A,B): P\mapsto Qu^{x} \ .
\end{displaymath}

\begin{figure}[ht]
\centering \includegraphics[width=0.4\textwidth]{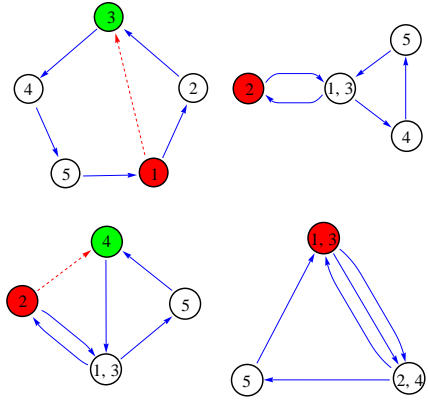}
\caption{\sf The action of the constructive AB maps on the polygonal directed quiver to a toric quiver. The red node is the current location of the walker, the left hand column represents the application of the contractions ($a_{i}$), from the red node to the green node. The right hand column shows the state after the application of the walks ($b_{i}$). The diagram shows the application of $A=\{2,2\},B=\{1,4\}$.}
\label{Condia}
\end{figure}

This AB notation has a number of useful symmetry properties which are due to A and B containing relative node positions on the quiver. 
Rotating the quiver amounts to a cyclic permutation of the elements of both A and B and reflecting the quiver diagram is also manifest, but in a more complicated way.
We refer the reader to Appendix \ref{App:QP} for more details on these symmetries.
It is important to note that the notation is much more efficient than an incidence matrix, requiring only $2x$ numbers to label a quiver. 
Using these concepts and this notation we can restrict the space of possible incidence matrices and produce an algorithm much more efficient than simple brute force.

\subsection{Warm-up: The $x=2$ Problem}
\label{Sec:TBP}
As an example for applying these above methods, let us classify, as a first example, the $x=2$ quiver diagrams, that is, all quiver diagrams formed by two contractions from some polygon. 
These are the two-term quiver theories studied in \cite{Hanany:2008gx} and are quite simple quivers with $A=\{a_{1},a_{2}\},B=\{b,E-b\}$.
With the above ideas, it should be possible to analyse this case easily. 
For simplicity only quivers with no self-adjoint loops are considered, as $x=2$ quivers with self-adjoint loops are trivial extensions of $x=0$ and $x=1$ quivers, i.e., we just add a loop to a single node for the $x=1$ quivers, or two loops for the $x=0$.

Each $x=2$ quiver is made by applying some contraction operation to an $x=1$ quiver, so we must consider the possible $x=1$ quiver diagrams first.
We need to work out which quivers are unique. 
One immediately sees that applying $C(a)$, where $a$ is the node spacing, will give us an $x=1$ quiver, however if $a>\text{floor}(E/2)$ we can flip this diagram to become a  quiver with $a' = a-\text{floor}(E/2)$. 
Therefore all the $x=1$ quivers are formed by applying the operation $C(a_{1}) \text{ for } a_{1}=2,3,...,\text{ floor}(E/2)$. 
To produce the $x=2$ quivers we first apply the walk operation $W(b_{1})$, to move to another node in the quiver. Then another contraction $C(a_{2})$ is applied . 
The values of $b_{1}$ and $a_{2}$ which can be applied to produce non-isomorphic quivers split into a number of regimes, based upon their values.

\begin{figure}[h]
\centering \includegraphics[width=0.5\textwidth]{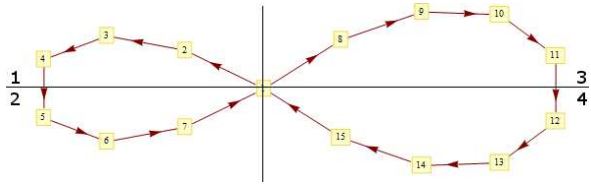}
\caption{{\sf The solution to the warm-up problem of $x=2$, by splitting up into regions 1 to 4.}}
\label{bprob}
\end{figure}
We present the procedure graphically in Figure \ref{bprob}, where we have labelled 4 regions.
The idea for solving is that we will systematically contract all nodes with each other, but single out those which will produce isomorphic quiver diagrams. We first look within the first quarter of the general $x=1$ diagram above, confining contractions to this quarter. We then match up all the contractions across quarters. As the quarters cut the diagram along potential symmetry lines, we can use this view of the diagram to easily work out which contractions will be isomorphic to one another. Due to the symmetry conditions, if we perform $C(2,6)$ on the above diagram, we have also performed $C(3,7)$. So we must factor this condition into our systematic categorisation of the possibilities.
 
The matchings we choose to do go as follows: 
In region 1 we match up with regions 1 and 2. 
In region 2 we match with region 3 (this is isomorphic to matching region 1 with 3 or 4, and region 2 with 4).
In region 3 we match with regions 3 and 4.

By working through these cases we will discover all the possible values of $AB$ which let us create quiver diagrams not isomorphic to one another. We also note the no self-adjoint loop condition, this has a number of effects. 
First, we restrict the members of $A$ to be strictly greater than one. 
Next, we remember in computing the conditions below that we do not want to create a self-adjoint loop by pulling any node which is one edge away from another together. 
Specifically, this can happen for $x>1$ quivers as applying contractions when there is a loop already in the diagram can lead to a situation where the $a$ of the contraction is greater than one, but the distance of the node from the multiply connected node is only one. For the case below this is easy to avoid. 
However, more complex cases would prove a very real problem.

First we have from the $x=1$ result that $2 \leq a_{1} < floor(E/2)$, and let $b_{m} = a + floor((E-a)/2)$ to simplify notation, then let $0 < b \leq b_{m}$, and so:

if $a_{1} \neq E/2$ then we find following conditions on $a_{2}$ for a given $b$ and $a_{1}$:
\begin{eqnarray}
\nn &\text{if } b = 0 \text{ then } 2 \leq a_{2} < floor(a_{1}/2)\\
\nn &\text{if } 0 < b < floor(a_{1}/2)\text{ then } 2 \leq a_{2} < a_{1}-2b\\ 
\nn &\text{if } floor(a_{1}/2) \leq b < a_{1}\text{ then } a_{1}-b+1 \leq a_{2} < b_{m}-b\\
\nn &\text{if } b = a_{1}\text{ then } 2 \leq a_{2} < b_{m}-b\\
&\text{if } b > a_{1}\text{ then } b_{m}-b \leq a_{2} < E-2b+a \ ;
\end{eqnarray}

if $a_{1} = E/2$ then we gain another line of symmetry and the only matchings required are:
region 1 with regions 1 and 2;
region 2 with region 3.
\begin{eqnarray}
\nn &\text{if } b = 0: 2 \leq a_{2} < floor(a_{1}/2)\\
\nn &\text{if } 0 < b < floor(a_{1}/2): 2 \leq a_{2} < a_{1}-2b\\ 
\label{noadj}
&\text{if } floor(a_{1}/2) \leq b < a_{1}: a-b+1 \leq a_{2} < b_{m} + a - b  \ .
\end{eqnarray}

Finally, we consider that we could have constructed isomorphic quivers which have the same loops but in the opposite order, i.e., with $a_{1}$ and $a_{2}$ swapped. In order to make sure this does not occur, we restrict the value of $a_{2}$, so that the loops always have a definite order depending on their size.
\begin{eqnarray}
\nn
&a_{2} \leq a_{1}\\
&a_{2} > 1 \ .
\end{eqnarray}

This exhaustively produces all possible quiver diagrams for $x=2$.

It would be incredibly complex to conduct the same kind of analysis for larger $x$. Due to this difficulty in implementing this kind of constructive solution for arbitrary $x$, the next step is to try an algorithmic solution, and use a graph isomorphism test to find quiver diagrams representative of their isomorphism class. 
To this we now turn.

\setall
\section{The Algorithm for Generating Toric Quiver Theories}\label{Sec:Al}
Our warm-up example clearly lends itself to an algorithmic approach.
In this section, we briefly outline this algorithmic procedure which we device to systematically generate the toric quivers and all possible associated superpotentials.
The method has been implemented in C and is presently being organised into a convenient software.
This section, together with the companion appenix, are rather technical and the readers interested only in the results can safely skip the ensuing text and move onto the next section where the results of our classification will be presented.

\begin{figure}[!ht]
\centering \includegraphics[width=0.35\textwidth]{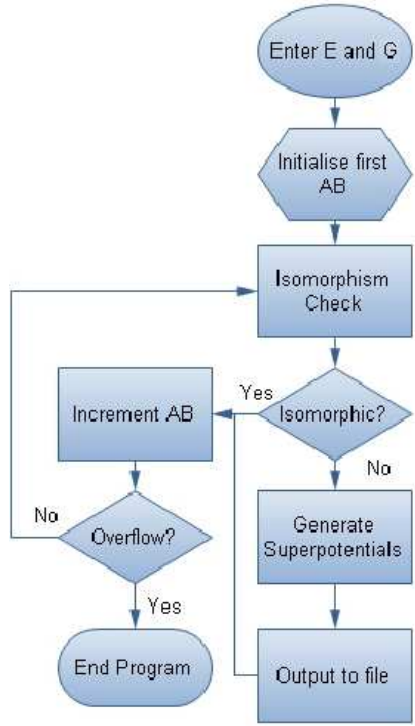}
\caption{\small Flowchart of the basic structure, showing the main processing sections and decisions of our algorithm for generating toric quiver theories.}
\label{QTflow}
\end{figure}

Our methodology is as follows.
We will generate all the toric quiver diagrams by finding all possible AB, and test them to pick a representative of each graph isomorphism class. 
As AB is represented by a list of numbers we treat the problem as though the AB are numbers in base $E$ that can be counted through;
digits greater than $E$ need not be considered as we would just loop around the whole diagram. 
Next, given a quiver diagram we wish to check if it is isomorphic to any quiver diagram which can be written as a ``smaller'' AB, that is, a smaller base $E$ number. 
The `smallest' AB is then kept as the representative of the class. 
With this quiver diagram we find the possible superpotentials using the method detailed in Section \ref{Sec:Theory}. 
The algorithm terminates when the AB, (interpreted as a base $E$ number), overflows. 
The flowchart Figure \ref{QTflow} sets this out in detail.

\subsection{Counting}
When implementing, the isomorphism test should check as few quivers as possible to maximise efficency. 
The AB can be constrained to reduce the number that must be checked. 
To do this, the so-called `\emph{moduli rules}' are used. 
These rules constrain the AB so the operation it represents never crosses the starting node and are given by imposing
\begin{displaymath}
\sum b_{i} + a_{i+1} < E \ .
\end{displaymath}
Moduli rules are used instead of restricting using symmetry since using symmetry leads to problems with redundancy and is more computationally complex. 
We can also place limits on the members of $A$ and $B$. 
We require all $a_{i} \geq 1$ since $C(0)$ is an identity and that $b_{i} \geq 1$ to eliminate redundancy; again see Appendix \ref{App:QP} for proofs of the above statements. 
With these rules it is a simple matter to write the code which increments the AB, treating it as a base $E$ number, and obeys these conditions. 
To implement the counting part of the algorithm a number of routines are used; these are detailed in the Counting section of Appendix \ref{App:alg}.

\subsection{Isomorphism}
A simple graph isomorphism test would require the current quiver diagram to be checked against all quiver diagrams found so far, so a method that can test a single possible quiver at a time is preferable. 
A solution is to construct all possible ways of writing the quiver diagram in AB notation. 
This is done by walking all Eulerian cycles and counting the edges between c-nodes in these cycles to find all the AB. 
It is basically the reverse process of collapsing the polygon with the AB; we refer the reader to appendix \ref{App:AT}. 
Walks need only start from c-nodes, since $C(a)$ is always applied first in an AB, increasing the efficiency. In fact, in Appendix \ref{App:AT} we show that the routine has complexity $O(E^{ax})$.
The``smallest'' way of writing the quiver diagram is then taken as the representative.
We detail the routines in the Isomorphism check section of Appendix \ref{App:alg}.

\subsection{Superpotential Generation}
In order to produce the superpotentials we first identify all c-nodes in the quiver diagram and then generate all edge assigments. 
For each edge assignment, we trace through the loops which it forms in the quiver, and then generate the gauge invariant terms. 
This is precisely the method discussed in section \ref{Sec:Theory}. 
As the order of fields within a term is irrelevant they are ordered in a consistent way to check if term collections share terms. 
Then the term collections are matched up to form superpotentials with the correct number of terms.
The sub--routines are explained in the Superpotential Generation Section of Appendix \ref{App:alg}.

\setall
\section{Classification Results}\label{Sec:Res}
In this section, we put the above theoretical algorithms to practice and present the complete quivers for some low values of the order parametres.
Many quivers have already eluded current methods of attack due to their complexity and our algorithm can generate them efficiently.
Moreover, we will initiate an exploration of this space, consisting of many quivers new to the literature, and see what structures we can unravel.  

\subsection{Quiver Diagrams for Low Order Parameters}
We now present some of the low order quiver diagrams for $x=0,1,2,3$. The $x=0$ and $x=1$ quiver diagrams are trivially an $E$ sided polygon and single contractions $A={a},B={E}$, so we need not draw them explicitly.
Moreover, we will for now focus on what we call {\bf base quiver diagrams}; that is, diagrams without self-adjoint loops.
As we mentioned above, these adjoints are not captured by the incidence matrix because they would be ambiguously represented by a column of zero. 
One can be put these in by analysing candidates for the superpotentials as was done in \cite{Hanany:2008gx}.

The $(G,E)$ equaling to $(3,5)$, $(4,6)$ and $(5,7)$ quiver diagrams are the lowest order $x=2$ diagrams which are not trivially found. 
We see from inequality \eqref{Eqn:LL} that if $x=2$, then $E \geq 3$. 
Requiring for now that there are diagrams with no self-adjoint loops takes this limit up to $E=5$. 
The $(3,5)$ and $(4,6)$ diagrams match those analysed by \cite{Hanany:2008gx}.

\begin{figure}[!ht]
\centering
\includegraphics[trim = 5mm 5mm 5mm 5mm, clip, width=0.40\textwidth]{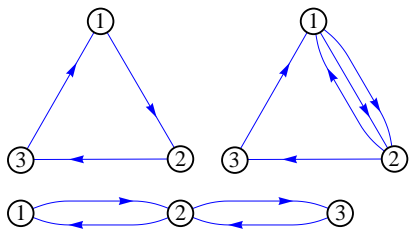}
\caption{{\sf $(G,E)=(3,5)$ Quiver Diagrams, self adjoint loops are implied in the graphs with edges fewer than $E=5$.}}
\label{Qu35}
\end{figure}
Diagrams with self adjoint loops are not shown below (except for the $(3,5)$ case) as they are the trivial $x=1$ and $x=0$ diagrams with one (or more) self adjoint loop(s) attached to some node(s). There are three $(3,5)$ base quiver diagrams shown in Figure \ref{Qu35}: a single quiver diagram with no self-adjoint loops, one with a single self-adjoint loop and one with two self-adjoint loops. There are five $(4,6)$ base quiver diagrams shown in Figure \ref{Qu46} (allowing for adjoints there would be 7). 
In Figure \ref{Qu57} the seven $(5,7)$ base quiver diagrams are presented (allowing adjoints there are ten $(5,7)$ quivers).

\begin{figure}[!ht]
\centering
 \includegraphics[trim = 20mm 20mm 20mm 20mm, clip, width=0.50\textwidth]{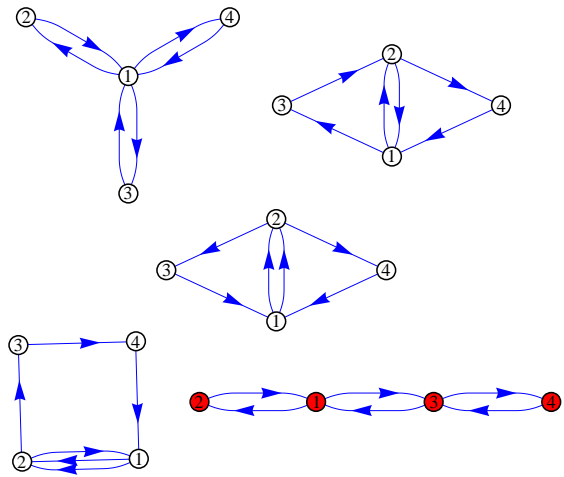}
\caption{{\sf $(G,E)=(4,6)$ base Quiver Diagrams, with no self adjoint loops. The diagram with red nodes does not admit a toric superpotential even if fields are not constrained to only take complex values.}}
\label{Qu46}
\end{figure}

\begin{figure}[!ht]
 \centering 
\includegraphics[trim = 32mm 32mm 20mm 20mm, clip, width=0.50\textwidth]{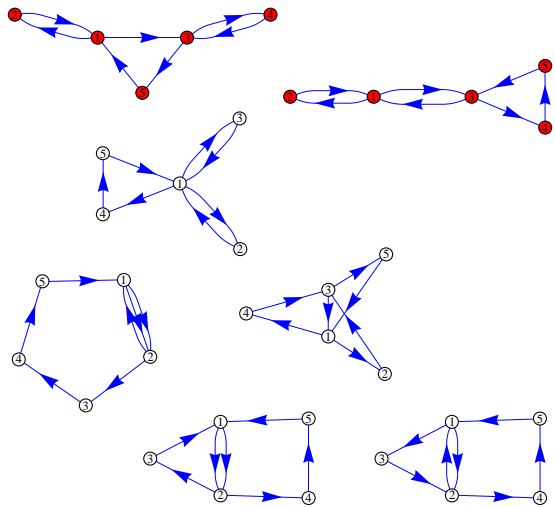}
\caption{{\sf All $(G,E)=(5,7)$ quiver diagrams with no self-adjoint loops. The diagrams with red nodes have vanishing toric superpotentials even if the fields are not constrained to only take complex values.}}
\label{Qu57}
\end{figure}

Next, let us move on to $x=3$, and present some low order complete sets of quiver diagrams. 
This has not appeared in the literature before.
We find that $E \geq 6$ from inequality (\ref{Eqn:LL}) and so present the three $(3,6)$ base quiver diagrams with no self adjoint loops in Figure \ref{Qu36}. 
The quivers with self-adjoint loops are those in Figure \ref{Qu35}, but with loops attached, giving six quiver diagrams in total. 
In Figure\ref{Qu47} the five $(4,7)$ quiver diagrams with no self-adjoint loops are presented. 
The single self adjoint loop base quiver diagrams have already been shown in Figure \ref{Qu46}.  
\begin{figure}[!ht]
\centering
\includegraphics[trim = 5mm 5mm 5mm 5mm, clip, width=0.40\textwidth]{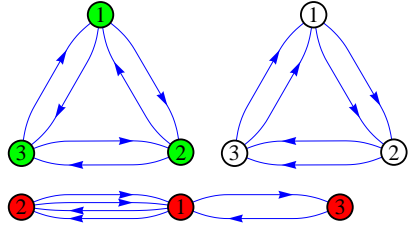}
\caption{{\sf $(G,E)=(3,6)$ Quiver Diagrams, with no self adjoint loops. The diagram with green loops has a single toric superpotential and the diagram with red nodes admits no superpotentials.}}
\label{Qu36}
\end{figure} 

\begin{figure}[!ht]
\centering
\includegraphics[trim = 10mm 20mm 20mm 20mm, clip, width=0.50\textwidth]{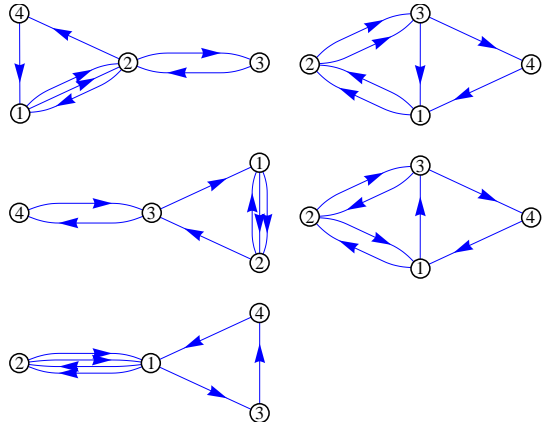}
\caption{{ \sf $(G,E)=(4,7)$ Quiver Diagrams, with no self adjoint loops.}}
\label{Qu47}
\end{figure}

To demonstrate that the algorithm is still efficient for higher values, let us move on to an example set for $x=4$, with $(G,E) = (4,8)$, which again is new to the literature. We present all base quivers in Figure \ref{Qu48} and there is a total of 16. From this list we recognise, among others, the 8-arrowed quiver for the cone over the zeroth Hirzebruch surface.
\begin{figure}[!ht]
\centering
\includegraphics[trim = 10mm 0mm 0mm 0mm, clip, width=6in]{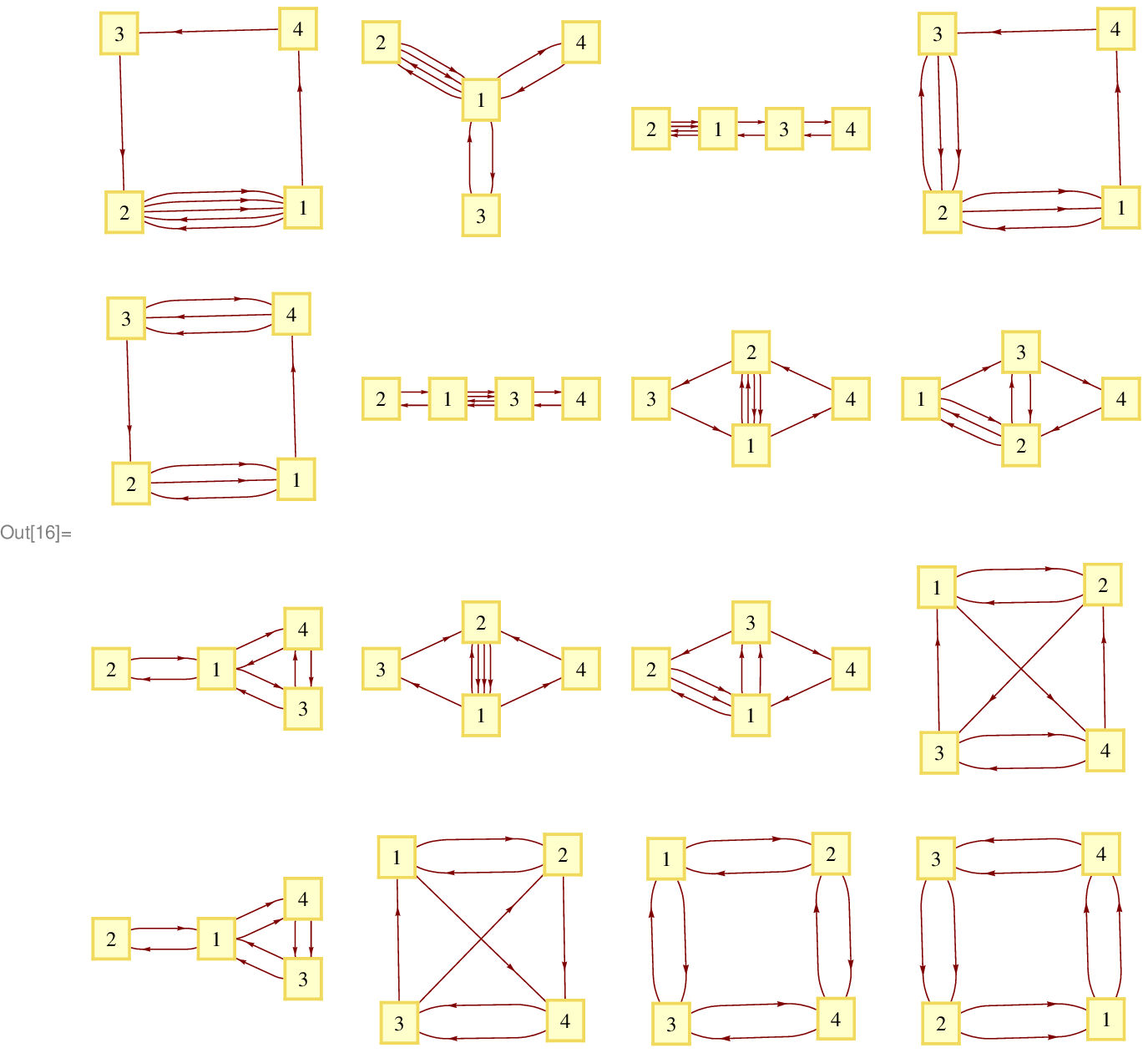}
\caption{{ \sf $(G,E)=(4,8)$ Quiver Diagrams, with no self adjoint loops.}}
\label{Qu48}
\end{figure}

\subsection{Assigning the Superpotential}
The above are simply quiver diagrams, providing any of them with a quiver representation gives a new physical theory with $E$ matter superfields and $E-x$ gauge groups. 
These classifications and the many others made by the algorithm for higher $E$ and $x$ therefore produce many new physical theories, the majority of which have not yet been studied and were previously unknown. It would certainly be interesting to study their VMS to identify the geometry.
Some, of course, have appeared in the literature in interesting and significant applications. The single $(2,2)$ quiver diagram produced by the program was first used as for the supersymmetric gauge theory associated to the conifold \cite{Klebanov:1998hh}. One of the $(3,6)$ quiver diagrams was one of the very early quiver theories, used as a model for a conformal supersymmetric gauge theory\cite{Douglas:1996sw}. These diagrams, among others, confirm the results found using the program and show that so far we have barely scratched the surface of the space of toric quiver theories.

To produce full TCYQTs we require toric superpotentials. 
The algorithm has found superpotentials for a large number of quiver diagrams. 
We again systematically work through the low order diagrams. 
The $x=2$ quivers require two terms in the superpotential. 
In the TCYQT case both the positive and negative terms contain all fields and so will cancel giving no toric superpotentials in the Abelian case of a single brane for any $x=2$ quiver diagram.
The superpotentials for the $(3,6)$ quivers are found to be identical to the example results in \ref{Sec:Theory}, as they should be. 
It is also found that the $(3,6)$ quiver diagram with green nodes in Figure \ref{Qu36} admits a single superpotential:
\begin{equation}
W = X_{12}^1X_{23}^1X_{31}^1+X_{12}^2X_{23}^2X_{31}^2 -X_{12}^1X_{23}^1X_{31}^1X_{12}^2X_{23}^2X_{31}^2 \ .
\end{equation}
The $(3,6)$ diagram with red nodes does not admit any superpotentials which do not vanish for the Abelian case.
Including the quivers with self-adjoint loops there are twelve $(3,6)$ TCYQTs in total.

\begin{figure}[!ht]
\centering
\includegraphics[trim = 5mm 5mm 5mm 5mm, clip, width=0.35\textwidth]{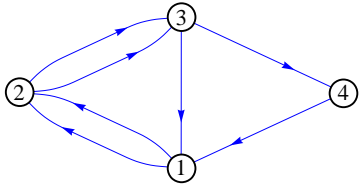}
\caption{{\sf An example of $(G,E)=(4,7)$ quiver diagram, for which we present all superpotentials in Equation \eqref{E7G4W}.}}
\label{Qu47ex}
\end{figure}
As a more detailed example we now stdy a single $(4,7)$ quiver diagram, presented in Figure \ref{Qu47ex}, which has no self-adjoint loops and admits multiple toric superpotentials:
\begin{eqnarray}
\nn
W_{1} &=&X_{12}^{(1)}X_{23}^{(2)}X_{31}^{(1)}+X_{12}^{(2)}X_{23}^{(1)}X_{34}^{(1)}X_{41}^{(1)}- T_{all} \\
\nn
W_{2} &=&X_{12}^{(1)}X_{23}^{(1)}X_{31}^{(1)}+X_{12}^{(2)}X_{23}^{(2)}X_{34}^{(1)}X_{41}^{(1)}- T_{all} \\
\nn
W_{3}&=&X_{12}^{(2)}X_{23}^{(1)}X_{31}^{(1)}+X_{12}^{(1)}X_{23}^{(2)}X_{34}^{(1)}X_{41}^{(1)}- T_{all} \\
\label{E7G4W}
W_{4}&=&X_{12}^{(2)}X_{23}^{(2)}X_{31}^{(1)}+X_{12}^{(1)}X_{23}^{(1)}X_{31}^{(2)}X_{41}^{(1)}- T_{all} \ ,
\end{eqnarray}
where $T_{all} := X_{12}^{(1)}X_{12}^{(2)}X_{21}^{(1)}X_{23}^{(1)}X_{31}^{(1)}X_{31}^{(2)}X_{41}^{(1)}$ is composed of a full Eulerian cycle.
These superpotentials together with the quiver diagram create all possible TCYQTs for this diagram, and so these theories encode all possible toric moduli spaces. 
Thus the program has produced a complete set of new physical theories for the above cases. Further complete sets of TCYQTs have also been produced for larger $x$ and $E$, giving many more quiver theories than previous known. 

\begin{figure}[!ht]
\centering
\includegraphics[trim = 1cm 0mm 0mm 0mm, clip, width=0.35\textwidth]{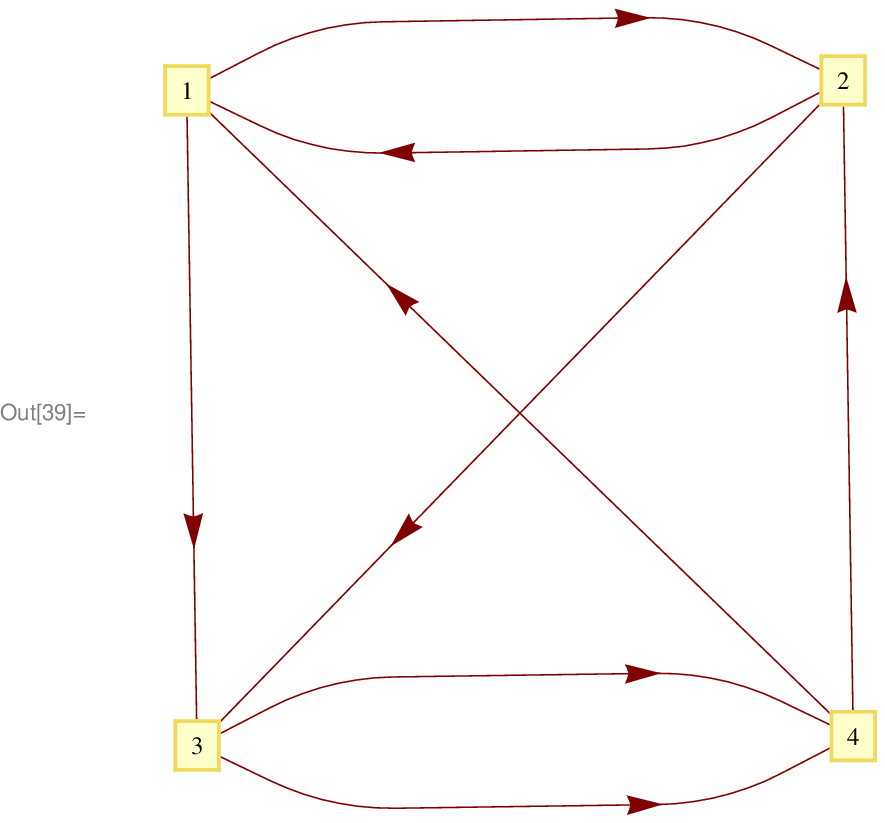}
\caption{{\sf An example of $(G,E)=(4,8)$ quiver diagram, for which we present all superpotentials in Equation \eqref{E8N4exW}.}}
\label{f:E8N4ex}
\end{figure}
As a final example, below are tables of the full sets of possible superpotentials for one of the sixteen $(4,8)$ diagrams above, drawn now in Figure \ref{f:E8N4ex}:
\begin{eqnarray}
\nn
W&=&
X_{12} X_{23} X_{34}^1 X_{41}-X_{13} X_{34}^2 X_{41}-X_{12} X_{21} X_{23} X_{34}^1 X_{42}+X_{1,3} X_{21} X_{34}^2 X_{42} \\ 
\nn
W &=&
 -X_{1,3} X_{34}^1 X_{41}+X_{12} X_{23} X_{34}^1 X_{41}+X_{1,3} X_{21} X_{34}^2 X_{42}-X_{12} X_{21} X_{23} X_{34}^2 X_{42} \\ \nn W &=&
 X_{12} X_{23} X_{34}^1 X_{41}-X_{12} X_{1,3} X_{21} X_{34}^2 X_{41}-X_{23} X_{34}^1 X_{42}+X_{1,3} X_{21} X_{34}^2 X_{42} \\ \nn W &=&
 -X_{12} X_{1,3} X_{21} X_{34}^1 X_{41}+X_{12} X_{23} X_{34}^1 X_{41}+X_{1,3} X_{21} X_{34}^2 X_{42}-X_{23} X_{34}^2 X_{42} \\ \nn W &=&
 X_{1,3} X_{34}^2 X_{41}-X_{12} X_{23} X_{34}^2 X_{41}-X_{1,3} X_{21} X_{34}^1 X_{42}+X_{12} X_{21} X_{23} X_{34}^1 X_{42} \\ \nn W &=&
 X_{1,3} X_{34}^2 X_{41}-X_{12} X_{1,3} X_{21} X_{34}^2 X_{41}+X_{12} X_{21} X_{23} X_{34}^1 X_{42}-X_{23} X_{34}^1 X_{42} \\ \nn W &=&
 -X_{12} X_{1,3} X_{21} X_{34}^1 X_{41}+X_{1,3} X_{34}^2 X_{41}+X_{12} X_{21} X_{23} X_{34}^1 X_{42}-X_{23} X_{34}^2 X_{42} \\ \nn W &=&
 -X_{1,3} X_{34}^1 X_{41}+X_{12} X_{23} X_{34}^2 X_{41}+X_{1,3} X_{21} X_{34}^1 X_{42}-X_{12} X_{21} X_{23} X_{34}^2 X_{42} \\ \nn W &=&
 -X_{12} X_{1,3} X_{21} X_{34}^2 X_{41}+X_{12} X_{23} X_{34}^2 X_{41}+X_{1,3} X_{21} X_{34}^1 X_{42}-X_{23} X_{34}^1 X_{42} \\ \nn W &=&
 -X_{12} X_{1,3} X_{21} X_{34}^1 X_{41}+X_{12} X_{23} X_{34}^2 X_{41}+X_{1,3} X_{21} X_{34}^1 X_{42}-X_{23} X_{34}^2 X_{42} \\ \nn W &=&
 X_{1,3} X_{34}^1 X_{41}-X_{12} X_{1,3} X_{21} X_{34}^2 X_{41}-X_{23} X_{34}^1 X_{42}+X_{12} X_{21} X_{23} X_{34}^2 X_{42} \\
\label{E8N4exW}
W &=&
 X_{1,3} X_{34}^1 X_{41}-X_{12} X_{1,3} X_{21} X_{34}^1 X_{41}+X_{12} X_{21} X_{23} X_{34}^2 X_{42}-X_{23} X_{34}^2 X_{42} \ .
\end{eqnarray}
We see that there are a total of 12 inequivalent possibilities.

\subsection{The Space of Quiver Theories}
\begin{table}[!htc]
\centering
\begin{tabular}[c]{l|lllllllll}
x & \multicolumn{8}{c}{E} \\
& 2 & 3 & 4 & 5 & 6 & 7 & 8 & 9 & 10 \\ \cline{2-10}
0 & 1/0 & 1/0 & 1/0 & 1/0 & 1/0 & 1/0 & 1/0 & 1/0 & 1/0 \\ \hline
1 & 0 & 2/0 & 2/0 & 3/0 & 3/0 & 3/0 & 4/0 & 4/0 & - \\ \hline
2 & & & 2/0 & 3/0 & 7/0 & 10/0 & 18/0 & 25/0 & 38/0 \\ \hline
3 & & & & 2/0 & 6/12 & 12/36 & 29/104 & 58/225 & 116/466 \\ \hline
4 & & & & & 3/0 & 8/47 & 28/330 & 72/1277 & 208/3941 \\ \hline
5 & & & & & & 3/0 & 12/0 & 44/645 & 172/4888 \\ \hline
6 & & & & & & & 4/- & 16/- & 84/- \\
\hline
\end{tabular}
\caption{{\sf Number of Quiver Diagrams/Quiver Theories with non-vanishing Abelian superpotential, plotted against $E$, the number of edges and $x = E-G$, the number of superpotential terms.}}
\label{Tb1}
\end{table} 

It is now expedient to summarise our results, which we plot in Table \ref{Tb1}, containing the total number of quiver diagrams and TCYQTs for a given $x = E-G$ and $E$.
It is seen that the lower bound on the edges for the possible existence of TCYQTs is obeyed by the numbers in the table, as well as the condition that $x \geq 3$ for there to exist a TCYQT having a non-vanishing Abelian superpotential.
To further confirm the results the numbers of $x=2$ quiver diagrams are compared with the generating function found in \cite{Hanany:2008gx}, which gives the number of base quiver diagrams that have at least one non-vanishing toric superpotential. The expanded function is:
\begin{equation}
g(x) = 1+2x+3x^2+6x^3+8x^4+13x^5 +17x^6+24x^7+\ldots \ ,
\end{equation}
with the number of $(N+2,N)$ quiver diagrams given by the coefficient of the $x^{(N-1)}$ term.

We see immediately that there seems to be a discrepancy with the $(4,6)$, $(5,7)$ terms and indeed the rest of the $x=2$ quivers. 
This is because \cite{Hanany:2008gx} does not restrict the fields to be complex valued and even for vanishing Abelian superpotential the gauge ranks could be promoted to $N>1$ and the two-term superpotential becomes non-vanishing.
This means that all diagrams in Figures \ref{Qu46} and \ref{Qu57} except the ones with red nodes admit a superpotential as the order of the fields in a superpotential term matters. 
The quiver diagrams with red nodes do not because an overall trace is taken on the superpotential, and the two terms in the superpotential correspond to the same Eulerian path and so are cyclically related. 
Therefore, one more quiver diagram is produced for the $(4,6)$ quivers and two more for the $(5,7)$ quivers and the results are then in perfect agreement. 
The number of quiver diagrams found must then be an upper bound on the generating function coefficient and we see that this holds for the new data. 

It is worthwhile to consider the trends of the numbers in Table \ref{Tb1}. 
Although there is limited data, there are a number of observations that can be made which provide evidence for the algorithm running correctly, and give more insight into the space of TCYQTs. 
The first is that large $x$ and small $E$ do not produce very many diagrams. 
This is fairly intuitive, as there are not many ways to place edges between a small number of nodes.
There are many more TCYQTs than quiver diagrams for high $E$ and $x$, but this is expected, as there should be many ways of decomposing a quiver diagram into loops, and consequently many term collections. 

One of the most interesting aspects is that whilst the number of diagrams increases monotonically with $E$ and $x$, the number of TCYQTs does not for low $E$. It appears the number of quiver diagrams increases faster with $E$, and the number of quiver TCYQTs increases faster with $x$, but it is hard to draw concrete conclusions from such a small data sample. 
However, we can at the very least say that the number of quivers becomes large quickly with $E$, and gets larger with $E$ faster for larger $x$.

As well as the production of the quiver theories, one of the main results of this paper is the algorithm itself. Each routine was tested individually\footnote{The testing strategy and more algorithmic details are detailed in Appendix \ref{App:AT}.}and that, together with the production of correct quiver diagrams and TCYQTs for low $x$ and $E$ provides sufficient evidence to be confident that the algorithm is robust.
By a rough complexity analysis, given in Appendix \ref{App:AT}, we find the algorithm has complexity $\sim O(E^{ax+b})$ for some positive constants $a$ and $b$, which for a given $x$ is polynomial in $E$. 
Examining the timing table in Appendix \ref{App:AT}, we see that this complexity roughly fits the data. 
Therefore it is feasible to use the program to produce quiver diagrams in a short time if the number of contractions is small, say $x \apprle 7$, and for large $x$ if time is not an issue. 

\setall
\section{Conclusions and Prospects}
\label{Sec:Con}
We have produced complete sets of toric quiver theories for low orders and a method to calculate all possible TCYQTs for higher orders. Many of the generated quivers have appeared in the literature for the first time, providing us with large number of new theories to study.

The idea of collapsing polygons to produce quivers and AB notation are new concepts. Moreover, the order parametre $x = E-G$, the number of terms in the superpotential, is also the number of contractions in this language.
 have allowed us to find a constructive solution to the $x=2$ problem, and, in general, an efficient algorithmic solution for arbitrary $x$. The AB notation is more efficient and more descriptive than the incidence matrix, in that it contains less redundant information, records the placement of self-adjoint loops and reflects the symmetries of the quiver diagram.

We have shown that the toric condition has a combinatoric interpretation of splitting the graph into disjoint loops, which we have exploited to generate superpotentials by constructing all sets of disjoint loops. It is easy to see from this that all toric spaces are Calabi-Yau spaces, as all Eulerian graphs can be split into disjoint loops, the proof of which we again leave to Appendix \ref{App:QP}.
The solution suggests that it is a better idea to characterise the order of quivers not by their edges and nodes $(E,G)$ as was in \cite{Hanany:2008gx}, but, rather, by the number of edges and the number of contractions $(E,x)$, because it is $x$ not and $G$ which determines the number of elements in the AB.

The computational complexity of the algorithm is roughly $O(E^{ax+b})$ for some positive constants $a$ and $b$, which is much better than the brute force case. 
It is unlikely that there exists an algorithm which has purely polynomial complexity in both $E$ and $x$, though it may be possible to improve the algorithm by using one of the more complicated graph isomorphism tests available.
Our plenitude of data can be used to study the space of general toric quiver theories more extensively than what has been possible before. 
As the algorithm is exhaustive, this data can be analysed systematically.
Producing all these types of quiver theories would involve expanding the algorithm to take matrix valued fields into account when producing superpotentials.

The toric quiver theories have been produced specifically so the `forward algorithm'\cite{Feng:2000mi} can be applied in order to scan through lower order theories, giving us a new method to study the duality manifested in the AdS/CFT correspondance. 
In particular, we should go through all the quivers together with their superpotentials and start listing the VMS by explicitly solving for the D-term and F-term equations.
The result should be all possible, stepwise, toric Calabi-Yau threefolds.
Likewise, we could add Chern-Simons labels to the nodes and follow the extended forward algorithm of \cite{Hanany:2008gx}. The results should be all toric Calabi-Yau fourfolds.
This is of tremendous interest because so far we do not have an inverse geometrical engineering method of constructing M2-brane world-volume gauge theories and exploring this space of theories under an algorithmic light should provide new insights and dualities.


\section*{Acknowledgements}
This work is based on the Masters in Philosophy project by J.~H.~under the supervision of Y.-H.~H.~and J.~H.~is grateful to University College, Oxford as well as the BP Thesis Prize of the Physics Department at the University of Oxford awarded for this project.
Y.-H.~H.~is dutifully indebted to an Advanced Fellowship from the STFC, UK, as well as the gracious patronage of Merton College, Oxford through the FitzJames Fellowship.

\bibliographystyle{unsrt}

\appendix
\appendixpage

\section{Some Technical Results on Quiver Constructions}
\label{App:QP}
In this appendix we collect some technical results which are used in the paper.
\subsection{Shortest Distance and the Adjacency Matrix}
There is a relation betwen shortest distances between two nodes and the adjacency matrix of the graph.
Consider the $n^{th}$ power of the adjacency matrix $A = a_{i}^{j}$. Using summantion notation, We see that $A^2 = a_{i}^{k}a_{k}^{j}$ and so 
\begin{equation}
A^n = a_{i}^{k_{1}}a_{k_{1}}^{k_{2}}...a_{k_{n-1}}^{k_{n}}a_{k_{n}}^{j} = (a_{i}^{j})^n \ .
\end{equation}
So, since the adjacency matrix is one if there is an edge between $i$ and $j$, and zero otherwise. Then $ (a_{i}^{j})^n $ is only non-zero if there is some path that goes from $i$ to $j$ in $n$ steps. As we require all the $a_{k_{p}}^{k_{p+1}}$ to be non-zero, simultaneously, in the sum. So there must be a path from $i$ to $j$ of $n$ steps in length if this element of $A^{n} \neq 0$. The first time this is non-zero will be the shortest path between $i$ and $j$, as this multiplication tries all routes. 

\subsection{A Quiver Admits an Eulerian Cycle}
As stated in the main text, every toric quiver admits an Eulerian cycle.
Here we present the proof, due originally to Fleury in \cite{hartmann-2005}.
\begin{proof}
Consider a $(E,G)$ quiver diagram, where $E-G = x$. 
Pick any node in the diagram. 
Now travel from this node, through the graph, never using any edge more than once. 
Every node we arrive at, we can then leave, apart from the node we started at. 
So we must end up back at the start, and have created a loop in the graph. 
Now remove that loop from the graph. The remaining graph is still a quiver, and we can perform the same operation. Continuing in this manner we can include every edge in some loop in the graph. 
Now reconstruct the whole graph, remembering the loops. Starting from the intial node, if we come to a node where another loop begins, we travel around it. In this way we travel through every edge once, and form an Eulerian cycle on the quiver. 
\end{proof}

\subsection{Cyclic Permutation of A and B}
There is a permutation symmetry on the A and B maps on the quiver.
Consider an $(E,x)$ quiver $A=\{a_{1}...a_{x}\},B=\{b_{1}...b_{x}\}$, and rotating this quiver to the next node which a contraction starts at. The first contraction is then $a_{2}$, and then next step, $b_{2}$, and so on until we reach $a_{x}$. Then the previous last step $b_{x}$, steps to the first contraction $a_{1}$ and $b_{1}$ steps to $a_{2}$. So a rotation of the quiver simply amounts to an identical cyclic permutation applied to both the A and B. 

\subsection{Mirror Symmetry of AB}
The mirror symmetry of a quiver is also reasonably nicely manifested. To construct the mirror image of a quiver, and remembering that we always start on a contracted node, we first contract $a_{1}$. Then we travel $b_{x}+a_{1}$ steps to the next contraction. We then contract $E-a_{x}$, and travel $b_{x-1}+a_{x}$ to the next contraction, then contract $E-a_{x-1}$ and so on. This gives a nice condition on the $a$'s. That the sum of the $a$'s is less than $Ex/2$, if we allow the $b$'s to take any value. Because of this allowance on the $b$'s the AB can 'wrap around' the quiver more than once which is why we cannot use this and the moduli rules together.

\subsection{Stopping Redundancy}
Setting $b \geq 1$ stops a redundancy when coupled with the moduli rules as in order to produce a redundant a then during the walk to collapse the polygon, we must arrive at a loop already formed, and then apply $C(a)$ where $a$ is the number of edges in the loop. If $b \geq 1$ we can never arrive at such a loop, as we can only create loops where we are, and then if we move forward once we cannot go back to the beginning of the loop. However, this doesn't stop us from making any possible quiver. As this c-node is now in front of us too, so we can contract any node in front of us with it, except for itself. In other words, this condition removes the redundant possibility whilst keeping all the ones that are not.

\section{The Algorithm}
\label{App:AT}

This appendix contains a brief set of documentation explaining how to test the algorithm. 
It also contains details of the implementation of the algorithm in C. 
The list below splits the program into its three sections, detailed in the flowchart in Figure \ref{QTflow}. 
It contains the information on what each of the implemented routines do, and how they were tested. 
Following this, results are shown for the runtime of the algorithm and the calculation of the approximate complexity algorithm is also shown.

\subsection{Implementation}
The isomorphism test works by creating a `path' object which walks through the quiver. 
First we generate a quiver diagram from the AB we are testing. 
We walk all possible Eulerian cycles of the quiver to find all possible AB. 
The walker starts at a c-node in the diagram, and then starts walking, every time the walker has a choice of where to go, it recognises it is at a c-node and so a contraction must begin from this node, it starts then counting from this node until it returns to it, at which point it has a value for that contraction. 
After leaving a c-node, the walker counts how many edges it walks until it reaches the next c-node. 
This is a B value. The program walks all paths simultaneously and so, when it reaches a c-node, where there is a choice, the walker copies itself, and walks both paths. 
This creates a web of Eulerian cycles across the quiver, each with a certain AB value. 
We then pick the smallest AB of these and compare it to the AB we used to generate the quiver. If it is the same, then we keep the AB, if it is smaller, we discard our tested AB. In this way, we find a representative of each graph isomorphism class and so generate all the quiver diagrams we require.

To output the data, a routine for converting the AB into an incidence matrix is needed. An incidence matrix for the original polygon is created, then a list of contracted $(i,j)$ is made. 
This array is then used to sum the correct rows from the polygon incidence matrix, equivalent to contracting the set of nodes stored in the row. We can then form the incidence matrix by taking each of these rows. The incidence matrix is the ouput to a file, in a format suitable for the toric analysis of the graphs.

To produce all the quiver diagrams, without the superpotential generation, a stripped down version of the program is used. 
This version does not contain the superpotential generation routine and has a very slightly different structure. 
As the superpotentials are not calculated, the self-adjoint loops become irrelevant, as they can go anywhere, and do not add any useful data to the diagram. 
So, a routine which checks for self-adjoint loops is used, which decreases the number of AB which have to be checked for isomorphism. This check occurs when the program creates an adjacency list, as it is easy to see if the loops are formed, as a node will be contracted with one next to it in the array.

To generate the stripped down diagrams, the ones missing self-adjoint loops, the algorithm is called again, but for $(N,E-1)$, then for $(N,E-2)$, until $E = N$. This has the added benefit of producing all the quivers with the same $N$ down to $(N,N)$ for free.
\subsection{Break of Sub-routines in Quiver Generation}
\label{App:alg}
The routines break down as follows:
\begin{enumerate}
\item Counting section:
   \begin{enumerate}
   \item Main - Accepts intial variables and sets up the problem, tested by printing the input variables to the screen, and making sure that the routine calls the following ones.
   \item Generate Quivers - Sets initial A and starts counting, tested by printing the intial A and checking that the final A is the first one which overflows.
   \item Aperm - Counts through A's, tested by printing each A to the screen and checking that they are being counted through correctly.
   \item Bperm - Counts through B's, tested by printing each B to the screen and checking that they are being counted through correctly. Due to the increased complexity of counting through the B, due to the varying end points for the different b's, each overflow was checked by hand for the small E and x. Then, a testing routing was used to check that the AB never broke the moduli rules for the other higher order AB.
   \item Incidencegen - Calls the isomorphism checking routine for the quiver, produces incidence matrix, and outputs to file, calls superpotential generation routine. Tested by checking the correct incidence matrix was produced for low complexity quivers, and it's easy to look in the file to check that the output is correct. 
   \item aadder - Handling routine, increments A, tested by printing the A before and after incrementation.
   \item badder - Handling routine, increments the first b of B. Tested by printing the B before and after.
   \item recurseadd - Handling routine, increments B, tested by printing the B to the screen.
\end{enumerate}

\item Isomorphism check section
\begin{enumerate}
   \item Finalcheck - Creates all possible paths, steps the paths through the graph by calling the stepping routing and calls the comparison routine. Tested by printing all the created paths to the screen, by printing that a path has been stepped to the screen and printing the result of the comparison.
   \item adjmaker - Creates the adjacency list for the given AB, tested by printing the adjacency list to the screen and comparing it to the corresponding graph in AB notation. 
   \item disconnectcheck - Checks to see if a path has made the graph disconnected, tested by printing the path, and the result of the disconnect check, and making sure that it was only giving the positive result for actually disconnected paths.
   \item pathstepper - Steps a path along to the next adjacent nodes, creates more paths if the node is multiply connected, tested by printing the path to the screen every step, and checking that the steps are working correctly.
   \item pathcompare - Compares all produced paths to the given AB, tested by printing the whole list of paths, to check their AB and comparing them by hand to the given AB.
   \item deletepath - Deletes a path from the pathlist, tested by printing the list before and after the path has been deleted.
   \item duplicatepath - Duplicates a path and puts it in the list before the path that was duplicated,tested by printing the list of paths before and after duplication.
   \item initialisepath - Initialises all the variables in an empty path, tested by printing the path after initialisation.
   \item quicksort - Handling routing - Sorts a list of data, tested by printing the list to be sorted before and after.
   \end{enumerate}

\item Superpotential Generation Section
   \begin{enumerate}
   \item Constructsuperpotentials - Calls generateterms, and constructs superpotentials from the produced term collections. Tested by printing all term collections and putting them together by hand to check they match up with the ones the program produced.
   \item cancelcheck - Checks to see if two term collections contain any like terms, tested by printing the list and making sure all term collections are different.
   \item printsuper - Prints the superpotential to a file, tested by reading the file.
   \item assignmentadder - Changes the edge assignments in a systematic way, tested by checking that all assignments are tested, by working out what they should be by hand, and printing the program generated ones to the screen.
   \item generateterms - Generates all the terms for a given edge assignment and adds them to the list of term collections, tested by taking some edge assignments, and running through the loops by hand, then comparing them to the term collections printed on the screen.
   \item identcheck - Checks whether any two term collections are the same, tested in the same way as cancelcheck.
   \item deletecol - Deletes a term collection from the list, tested by checking the pointer to the term collection is made null after the term collection is deleted.
   \item termsort - Takes a term collection and organises it so the numbers of fields, and the entries in a term are in order, tested by printing the terms before and after and checking that they are put in order.
   \item tquick - Does the sorting of the terms, tested in the course of testing termsort.
   \item fquick - Does the sorting of the fields, tested in the course of testing termsort.
   \end{enumerate}
\end{enumerate}

Each routine was tested to check that it produces the output that it is supposed to, and routines were tested in such a way that any dependencies were tested before the routine which was dependant on them. 
Testing the counting section was a simple matter, as that had to be done was to check that the program was indeed counting through all of the potential quivers, which we could look at the screen and see. The isomorphism check section was tested by dumping the contents of given path structures to the screen, so we could check that the adjacency list and pathstepper were working properly. The disconnectcheck was tested by following paths which became disconnected and checking that the checker picked them up. Superpotential generation was tested by dumping the contents of given term collections to the screen and counting the number of edge assignments that should be counted through.

Obviously the ideas behind the algorithm work, and we have proven this earlier on in the main text. As a test, we applied the algorithm, by hand, to a number of sets of low order quivers, and compared the results with the ones the program gave. The results it gave were indeed the same. 
However, it is hard to probe any further as it becomes incredibly laborious, and a person is more likely to make mistakes than a computer is at that stage. 
The algorithm itself however is quite robust and quite modular. 

We were also able to compare the results of the program in the $x=2$ case to the direct implementation of the solution to the warm-up problem and check that both programs produce the same results, which they do. 
To give a better idea of exactly how some of the routines work, the isomorphism test and output routine are covered in more detail below.

\subsection{Computational Results}
To analyse the efficiency of the program, we look at the below table, which records the time taken for the program to run. 
We see that the program is ``fast'' in $E$, and ``slow'' in x. 
This fits with our analysis of the complexity of the algorithm below.

\begin{table}[h!t]
\caption{Table showing time (in seconds) for Quiver Diagrams/Quiver Theories generation}
\centering
\begin{tabular}[c]{l|lllll}
E & \multicolumn{5}{c}{x} \\
& 2 & 3 & 4 & 5 & 6\\ \cline{2-6}
4 & 1/- & & & & \\ \hline
5 & 2/- & 1/1 & & & \\ \hline
6 & 2/- & 2/5 & 4/4 & & \\ \hline
7 & 2/- & 2/5 & 2/7 & 11/34 & \\ \hline
8 & 3/- & 3/6 & 7/25 & 9/151 & 15/- \\ \hline
9 & 3/- & 8/15 & 27/117 & 63/908 & 16/- \\ \hline
10 & 4/-&24/45 & 147/571 & 835/6690 & 1895/- \\ \hline
\end{tabular}
\end{table} 

Simple rules to note are that loops are polynomial in the number of entries in the loop, and that counting through numbers is exponential in the number of digits. 
This is a very rough analysis and should just give an idea of how efficient this algorithm is, i.e., whether it is fast or slow. 

We see that the counting section of the program is exponential in its complexity in $x$. 
As it counts through quivers the order of magnitude of the number of quivers it has to count through is roughly $O(E^x)$. 
There is, as far as we can see, no way around this as every quiver must be checked. 
The isomorphism checker has complexity $O(E2^{x})$\footnote{As discussed in \cite{GIsoDisease}. This seemingly amounts to a brute force method, however, because we only start paths on the c-nodes the method does in fact scale polynomially with $E$, as we have $2^{x}$ paths maximum, which we step $E$ times. Hence the number of steps is $E2^{x}$}, as has already been discussed.
The superpotential generator has to count through all assignments and generate the terms for each one. Term generation is polynomial time in $E$ as it simply loops through all the edges and records the term collections. However, the assignments are again exponential in $x$, so we get a complexity like $O(E^{ax+b})$. 
\end{document}